\documentclass[12pt,draftcls,onecolumn]{IEEEtran}

\usepackage{amssymb,latexsym,amsmath}
\usepackage{epsfig}
\usepackage[letterpaper,margin=1in]{geometry}
\usepackage[english]{babel}
\usepackage[latin1]{inputenc}
\usepackage[T1]{fontenc}
\usepackage{amsthm}
\usepackage{amstext}
\usepackage{color}%
\usepackage{mathrsfs}
\usepackage{mathtools}
\usepackage{setspace}
\usepackage{amsfonts}
\usepackage{amsmath}
\usepackage{amssymb}
\usepackage{amsthm}
\usepackage{bigints}
\usepackage{graphicx}
\usepackage{marvosym}
\usepackage{mathtools}
\usepackage{setspace}
\usepackage[enableskew]{youngtab}
\usepackage{fancyhdr}
\usepackage{tikz}					
\usetikzlibrary{arrows}

\def\R{\mathbb{R}}

\def\P{{\mathcal P}}
\def\Q{{\mathcal Q}}

\def\sA{{\mathsf A}}
\def\sH{{\mathsf H}}

\def\sM{{\mathsf M}}

\def\sX{{\mathsf X}}

\def\sZ{{\mathsf Z}}
\def\heta{\hat{\eta}}

\def\hX{\hat{X}}
\def\hx{\hat{x}}

\fussy

\theoremstyle{plain}
\newtheorem{thm}{Theorem}
\newtheorem{lem}{Lemma}

\newtheorem{proposition}{Proposition}

\theoremstyle{definition}
\newtheorem{defn}{Definition}

\theoremstyle{remark}

\renewcommand{\IEEEQED}{\IEEEQEDopen}

\allowdisplaybreaks

\title{Optimal Zero Delay Coding of Markov Sources: Stationary  and
  Finite Memory Codes}  

\author{Richard G. Wood, Tam\'as Linder, and Serdar Y\"uksel
\thanks{The authors are with the Department of Mathematics and
    Statistics, Queen's University, Kingston, Ontario, Canada, K7L
    3N6.  Email: richard.wood@queensu.ca, linder@mast.queensu.ca, yuksel@mast.queensu.ca}
\thanks{This research was
    partially supported by the Natural Sciences and Engineering
    Research Council of Canada (NSERC).}
\thanks{The material in this paper was presented in part at the
IEEE International Symposium on Information Theory, Hong Kong, Jun.\ 2015.}
}

\begin{document}

\maketitle

\begin{abstract}
  The optimal zero delay coding of a finite state Markov source is
  considered. The existence and structure of optimal codes are studied using a
  stochastic control formulation. Prior results in the literature established
  the optimality of deterministic Markov (Walrand-Varaiya type) coding policies
  for the finite time horizon problem, and the optimality of both deterministic
  nonstationary and randomized stationary policies for the infinite time horizon
  problem. Our main result here shows that for any irreducible and aperiodic
  Markov source with a finite alphabet, \emph{deterministic and stationary}
  Markov coding policies are optimal for the infinite horizon problem. In
  addition, the finite blocklength (time horizon) performance on an optimal
  (stationary and Markov) coding policy is shown to approach the infinite time
  horizon optimum at a rate $O(1/T)$.  The results are extended to systems where
  zero delay communication takes place across a noisy channel with noiseless
  feedback.
\end{abstract}

\emph{Keywords:} Zero delay source coding, real time coding, causal coding, quantization,
stochastic control, Markov sources,  Markov decision processes. 

\section{Introduction}

This paper is concerned with optimal zero delay coding of Markov sources for
infinite time horizons. Zero delay coding is a variant of the original lossy source
coding problem introduced by Shannon \cite{shannon1948}.

\subsection{Block Coding and Zero Delay Coding}\label{sec:problemdefinition}

Recall  Shannon's lossy source coding problem \cite{Cover}:
Given is an $\sX$-valued information source $\{X_t\}_{t\ge 0}$, where we assume
 that $\sX$ is a finite set. An encoder compresses the source at a rate $R$ bits
 per source symbol. A
decoder reproduces the information source via the sequence $\{\hX_t\}_{t\ge 0}$
of $\hat{\sX}$-valued random variables, where  $\hat{\sX}$ is also a finite set.   One is
typically concerned with the transmission \emph{rate} 
and the \emph{distortion} of the system.

In particular, a \emph{$(2^{RT},T)$-rate distortion block code} \cite{Cover}
encodes $T$ source symbols $X_{[0,T-1]}\coloneqq (X_0,\dots,X_{T-1})$ at a time,
and comprises an encoding function $\eta^T: \sX^T \to \{1,\dots,2^{RT}\}$ and
a decoding function $\gamma^T  : \{1,\dots,2^{RT}\} \to \hat{\sX}^T$.
This code has  rate $R$ bits per source symbol, and (expected) distortion given by
\begin{equation*}
D_T \coloneqq \frac{1}{T} E \left[ \sum_{t=0}^{T-1} d\left(X_t, \hX_t\right) \right],
\end{equation*}
where $(\hX_0,\dots,\hX_{T-1}) = \gamma^T(\eta^T(X_{[0,T-1]}))$ and  
$d:\sX \times\hat{\sX} \to [0,\infty)$ is a so called single letter distortion
measure.

A rate distortion pair $(R,D)$ is said to be \emph{achievable} if there exists a sequence of $(2^{RT},T)$-rate distortion codes $(\eta^T,\gamma^T)$ such that
\begin{equation}
\label{eq:ach}
\limsup_{T \to \infty} D_T \leq D.
\end{equation}

By a classical result, if the source is stationary and ergodic, the minimum
achievable distortion for rate $R$ is given by the distortion rate function of
the source
\begin{equation}
  \label{eq:r-d}
D(R)  = \lim_{T\to \infty} D_T(R),
\end{equation}
where $D_T(R)$ is the $T$th order distortion rate function which can be
calculated from the statistics of the block $X_{[0,T-1]}$ (see, e.g.,
\cite{ber71}).

As is evident from the definition of block codes, such a coding scheme relies on
encoding blocks of data $(X_0,\dots,X_{T-1})$ together, which may not be
practical for many applications as the encoder has to wait until it has all $T$
source symbols before it can start encoding and transmitting the data.  In
\emph{zero delay source coding}, the encoder can produce the code of $\hX_t$ as
soon as the source symbol $X_t$ is available. Such coding schemes have many
practical applications in emerging fields such as networked control systems (see
\cite{YukselBasarBook} and references therein for an extensive review and
discussion of applications), real-time mobile audio-video systems (as in
streaming systems \cite{Draper} \cite{Khisti}), and real-time sensor networks
\cite{Akyildiz}, among other areas.

In this paper, we consider a zero delay (sequential) encoding problem where the
goal is to encode an observed information source without delay. It is assumed
that the information source $\{X_t\}_{t\ge 0}$ is an $\sX$-valued discrete time
Markov process, where $\sX$ is a finite set. The transition probability matrix
$P$ and initial distribution $\pi_0$ for $X_0$ completely determine the process
distribution, so we will use the shorthand $\{X_t\}\sim (\pi_0,P)$.  The encoder
encodes (quantizes) the source samples and transmits the encoded versions to a
receiver over a discrete noiseless channel with common input and output alphabet
$\sM\coloneqq\{1,2,\ldots,M\}$, where $M$ is a positive integer.

In the following, we build on the notation  in 
\cite{YukLinZeroDelay}. Formally, the encoder is specified by a
\emph{quantization policy} $\Pi$, which is a sequence of functions
$\{\eta_t\}_{t\ge 0}$ with $\eta_t: \sM^t \times \sX^{t+1} \to \sM$.  At
time $t$, the encoder transmits the $\sM$-valued message
\[
  q_t=\eta_t(I_t)
\]
with $I_0=X_0$, $I_t=( q_{[0,t-1]} ,X_{[0,t]})$ for $t \geq 1$, where we have
used the notation $q_{[0,t-1]}=(q_0,\ldots,q_{t-1})$ and $X_{[0,t]} =
(X_0,\ldots,X_t)$.  The collection of all such zero delay encoding policies
is called the set of \emph{admissible} quantization policies and is denoted by $\Pi_A$.

Observe that for fixed $q_{[0,t-1]}$ and $X_{[0,t-1]}$, as a function of $X_t$,
the encoder $\eta_t(q_{[0,t-1]},X_{[0,t-1]},\,\cdot\,)$ is a \emph{quantizer},
i.e., a mapping of $\sX$ into the finite set~$\sM$. Thus a quantization 
policy at each  time index $t$ selects a quantizer $Q_t:\sX \to \sM$ based
on past information $(q_{[0,t-1]},X_{[0,t-1]})$, and then ``quantizes'' $X_t$ as
$q_t=Q_t(X_t)$.

Upon receiving $q_t$, the decoder
generates the reconstruction $\hX_t$, also without delay. A zero delay 
decoder policy is a sequence of functions
$\gamma=\{\gamma_t\}_{t\ge 0}$ of type  $\gamma_t : \sM^{t+1} \to
\hat{\sX}$, where  $\hat{\sX}$ denotes the finite reconstruction alphabet.  Thus
for all $t\ge0$,
\[
\hX_t=\gamma_t(q_{[0,t]}). 
\]

For the finite horizon (blocklength) setting the goal is to minimize the average
cumulative distortion (cost)
\begin{equation}\label{Cost1}
E^{\Pi,\gamma}_{\pi_0}\biggl[\frac{1}{T} \sum_{t=0}^{T-1} d(X_t,\hX_t)\biggr]
\end{equation}
for some $T \ge 1$, where $d: \sX \times \hat{\sX} \to [0,\infty)$ is a cost (distortion) function and $E^{\Pi,\gamma}_{\pi_0}$
denotes expectation with initial distribution $\pi_0$ for $X_0$ and under the
quantization policy $\Pi$ and receiver policy $\gamma$. We assume that
the encoder and decoder know the initial distribution $\pi_0$.

Since the source alphabet is finite, for any encoder policy $\Pi\in \Pi_A$ and
any $t\ge 0$, there always exists
an optimal receiver policy $\gamma^*=\gamma^*(\Pi)$ such that  for all
  $t\ge 0$,
\[
  E^{\Pi,\gamma^*}_{\pi_0}\bigl[d(X_t,\hX_t)\bigr] =
  \inf_{\gamma} E^{\Pi,\gamma}_{\pi_0}\bigl[ d(X_t,\hX_t)\bigr].
\]
From now on, we always assume that an optimal receiver policy is used for a given
encoder policy and, with an abuse of notation,  $\Pi\in \Pi_A$ will mean the
combined  encoder and decoder policies $(\Pi,\gamma^*(\Pi))$. Using this new
notation, we have for all $t\ge 0$, 
\[
   E^{\Pi}_{\pi_0}\bigl[d(X_t,\hX_t)\bigr]  =
  \inf_{\gamma} E^{\Pi,\gamma}_{\pi_0}\bigl[ d(X_t,\hX_t)\bigr].
\]

In this paper, we concentrate on the following performance criteria.

\begin{enumerate}

\item \emph{Infinite Horizon Discounted Cost Problem}:
In the infinite horizon discounted cost problem, the goal is to minimize the
cumulative ``discounted'' cost 
\begin{equation}\label{cost22}
J^{\beta}_{\pi_0}(\Pi) \coloneqq \lim_{T \to \infty} E^{\Pi}_{\pi_0}\left[\, \sum_{t=0}^{T-1} \beta^t d(X_t,\hX_t)\right]
\end{equation}
for some $\beta \in (0,1)$.

\item \emph{Infinite Horizon Average Cost Problem}:
The more challenging infinite horizon average cost problem has the objective of
minimizing the long term average distortion 
\begin{equation}
J_{\pi_0}(\Pi) \coloneqq  \limsup_{T \to \infty}
E^{\Pi}_{\pi_0}\left[\frac{1}{T}   \sum_{t=0}^{T-1}
  d(X_t,\hX_t)\right]. \label{infiniteCost}
\end{equation}

\end{enumerate}

We note that in source coding only the average cost problem is of interest, but
we also consider the discounted cost problem since it will serve as a
useful tool in studying the more difficult average cost problem. 

Observe that $R=\log_2 M$ is the rate of the described zero delay codes. Then,
in analogy to \eqref{eq:ach}, the rate distortion pair $(R,D)$ is said to be
achievable if there exists a policy $\Pi$ such that $J_{\pi_0}(\Pi) \le D$. As
opposed to the block coding case, finding the minimum achievable distortion
(cost) $\min_{\Pi\in \Pi_A} J_{\pi_0}(\Pi)$ at rate $R$ for zero delay codes is
an open problem. In particular, if the source is stationary and memoryless, then
this minimum is equal to $\min_f E[d(X_0,f(X_0))]$, where the minimum is taken
over all ``memoryless quantizers'' $f:\sX\to \hat{\sX}$ with $|f(\sX)| \le
2^{R}$ \cite{Eri79,GaSl79,GaSl82}. However, this optimum performance is not
known for any other (more general) source classes, and in particular it is
unknown when $\{X_t\}$ is a stationary and ergodic Markov source.  (Some partial
results on this problem are given in, e.g., \cite{GaSl82,GaGy86}.) 

Our main goal in this paper is to characterize some important properties of
optimal coding policies that achieve this minimum, even though we cannot
characterize the value of the minimum.

We review  two results fundamental to the structure of optimal zero delay codes
(see also \cite{YukIT2010arXiv}).

\begin{thm}[Witsenhausen \cite{Witsenhausen}] \label{witsenhausenTheorem} For
  the problem of coding a Markov source over a finite time horizon $T$, any zero
  delay quantization policy $\Pi=\{\eta_t\}$ can be replaced, without loss in
  distortion performance, by a policy $\hat{\Pi}=\{\heta_t\}$ which only uses
  $q_{[0,t-1]}$ and $X_t$ to generate $q_t$, i.e., such that
  $q_t=\heta_t(q_{[0,t-1]},X_t)$ for all $t=1,\ldots,T-1$.
\end{thm}

Let ${\cal P}(\sX)$ denote the space of probability measures on $\sX$. Given a
quantization policy $\Pi$, for all $t\ge 1$ let $\pi_t \in {\cal P}(\sX)$ be the
conditional probability defined by
 \[
\pi_t(A)\coloneqq \Pr(X_t\in A | q_{[0,t-1]})
\]
for any set $A\subset \sX$.

\begin{thm}[Walrand and Varaiya
  \cite{WalrandVaraiya}]\label{WalrandVaraiyaTheorem} For the problem of coding 
  a Markov source over a finite time horizon $T$, any zero delay quantization policy
  can be replaced, without loss in performance, by a policy which at any time
  $t= 1, \ldots,T-1$ only uses the conditional probability measure
  $\pi_t=P(dx_t|q_{[0,t-1]})$ and the state $X_t$ to generate $q_t$. In other
  words, at time $t$ such a policy $\heta_t$ uses $\pi_t$ to select a quantizer
  $Q_t=\heta(\pi_t)$ (where $Q_t:\sX \to \sM$), and then $q_t$ is generated as
  $q_t=Q_t(x_t)$.
\end{thm}

As discussed in \cite{YukIT2010arXiv}, the main difference between the two
structural results above is the following: in the setup of
Theorem~\ref{witsenhausenTheorem}, the encoder's memory space is not fixed and
keeps expanding as the encoding block length $T$ increases. In the setup of
Theorem~\ref{WalrandVaraiyaTheorem}, the memory space $\mathcal{P}(\sX)$ of an
optimal encoder is fixed (but of course is not finite). More importantly, the
setup of Theorem~\ref{WalrandVaraiyaTheorem} allows one to apply the powerful
theory of Markov decision processes on fixed state and action spaces, thus
greatly facilitating the analysis.

Recall that a Markov chain $\{X_t\}$ with finite state space $\sX$ is
\emph{irreducible} if for any $a,b \in \sX$ there exists a positive $n$ such
that $\Pr(X_n=b|X_0=a)>0$ (e.g., \cite[Chapter~1.2]{Nor97}), and it is
\emph{aperiodic} if for each state $a\in \sX$ there is a positive $n$ such that
$\Pr(X_{n'}=a | X_0 = a) >0$ for all $n' \geq n$ (e.g.,
\cite[Chapter~1.8]{Nor97}). Our assumption on the source $\{X_t\}$ is that it is
an irreducible and aperiodic finite state Markov chain.

The main  results in this paper are the following. 

\begin{itemize}
\item For the problem of zero delay source coding of an irreducible and
  aperiodic Markov source over an infinite time horizon we show the optimality
  (among all admissible policies) of {\it deterministic} and {\it stationary}
  (i.e., time invariant)  Markov (Walrand-Varaiya type) policies for both
  stationary and nonstationary Markov sources. 

\item For the same class of Markov sources, we show  that the
  optimum performance for time horizon $T$ converges to the optimum infinite
  horizon performance at least as fast as $O\big(\frac{1}{T}\big)$. 

\item Using the above convergence rate result, for 
  stationary Markov sources we also show the existence of $\epsilon$-optimal
  periodic zero delay codes with an explicit bound on the relationship between
  $\epsilon$ and the period length. This result is relevant since the complexity
  of the code is directly related to the length of the period (memory size).

\end{itemize}

The rest of the paper is organized as follows. In the next subsection we review
some existing results on zero delay coding and related problems.  In
Section~\ref{chap:finite} we derive auxiliary results and show that stationary
Walrand-Varaiya type policies are optimal in the set of all policies for the
infinite horizon discounted cost problem.  In Section~\ref{chap:avgcost} we
consider the infinite horizon average cost problem and prove the optimality of
stationary and deterministic Walrand-Varaiya type policies. The convergence rate
result and the $\epsilon$-optimality of finite memory policies are also
presented here.  Section~\ref{secnoisy} describes the extension of these results
for zero delay coding over a noisy channel with feedback.  Concluding remarks
are given in Section \ref{chap:conclu}.  In the Appendix we provide a brief
summary of some definitions and results we need from the theory of Markov
decision processes.

\subsection{Literature Review}
\label{sec_structure}

Structural results for the finite horizon coding problem have been developed in
a number of important papers. As mentioned before, the classic works by
Witsenhausen \cite{Witsenhausen} and Walrand and Varaiya \cite{WalrandVaraiya},
which use two different approaches, are of particular relevance.  An extension
to the more general setting of non feedback communication was given  by
Teneketzis \cite{Teneketzis}, and \cite{YukIT2010arXiv} also extended these
results to more general state spaces; see also \cite{YukLinZeroDelay} and
\cite{YukselBasarBook} for a more detailed review.

A related lossy coding procedure was introduced by Neuhoff and Gilbert
\cite{NeuhoffGilbert}, which they called {\it causal source coding}. The main
result in \cite{NeuhoffGilbert} established that for stationary memoryless
sources, an optimal causal coder can either be replaced by one that time shares
two memoryless coders, without loss in performance.  As noted in
\cite{NeuhoffGilbert}, zero delay codes form a special subclass of causal codes.
We also note that \emph{scalar quantization} is a practical (but in general
suboptimal) method for zero delay coding of continuous sources. A detailed
review of classical results on scalar and vector quantization is given in
\cite{GrNe98}.

Causal coding under a high rate assumption for stationary sources and individual
sequences was studied in \cite{LinderZamir}.  Borkar et al.\
\cite{BorkarMitterTatikonda} studied the related problem of coding a partially
observed Markov source and obtained existence results for dynamic vector
quantizers in the infinite horizon setting.  It should be noted that in
\cite{BorkarMitterTatikonda} the set of admissible quantizers was restricted to
the set of nearest neighbor quantizers, and other conditions were placed on the
dynamics of the system; furthermore the proof technique used in
\cite{BorkarMitterTatikonda} relies on the fact that the source is partially
observed unlike the setup we consider here.

In \cite{YukLinZeroDelay}, zero delay coding of $\R^d$-valued Markov sources was
considered. In particular, \cite{YukLinZeroDelay} established the existence of
optimal quantizers (having convex codecells) for finite horizons and the
existence of optimal {\it deterministic nonstationary} or {\it randomized
  stationary} policies for stationary Markov sources over infinite horizons, but 
the optimality of {\it stationary and deterministic} codes  was left as an open
problem.  Related work include \cite{AsnaniWeissman} which considered the coding
of discrete independent and identically distributed (i.i.d$.$) sources with
limited lookahead using the average cost optimality equation.  Also,
\cite{javidi2013dynamic} studied real time joint source-channel coding of a
discrete Markov source over a discrete memoryless channel with feedback under a
similar average cost formulation.

Some partial, but interesting results on the optimum performance of zero-delay
coding over a noisy channel are available in the literature. It is shown in
\cite[Theorem 3]{WalrandVaraiya} that when the source and the channel alphabets
have the same cardinality and the channel satisfies certain symmetry conditions
(e.g., the channel is the binary symmetric channel or a binary erasure channel),
then memoryless encoding is optimal for any Markov source. Also, an information
theoretic source-channel {\it matching} type argument can be made for special
scenarios where the sequential rate-distortion \cite{TatikondaThesis}
\cite{TatikondaSahaiMitter} achieving channel kernels are realized with the
physical channel itself, a crucial case being the scalar Gaussian source
transmitted over a scalar Gaussian channel under power constraints at the
encoder \cite{banbas89}. Along this direction, a more modern treatment and
further results are given in \cite{charalambous2014nonanticipative} and
\cite{KoChBo15}. Optimal zero delay coding of Markov sources over noisy channels
without feedback was considered in \cite{Teneketzis} and \cite{MahTen09}.

In this paper we  also investigate how fast the optimum finite
blocklength (time horizon) distortion converges to the optimum
(infinite horizon) distortion. An analog of this problem in block coding is the
speed of convergence of the finite block length encoding performance to
Shannon's distortion rate function. For stationary and memoryless sources, this
speed of convergence was shown to be of the type $O\big(\frac{\log T}{T}\big)$
\cite{pilc1967coding}, \cite{zhang1997redundancy}. See also
\cite{kostina2012fixed} for a detailed literature review and further finite
blocklength performance bounds.

Finally, we note that control theoretic tools are playing an increasingly
important role in solving certain types of problems in information
theory. Several of the papers cited above use dynamic programming as a crucial
tool to analyze an average cost optimal control problem that the given
information theoretic problem is reduced to. To facilitate this analysis, the
\emph{convex analytic method} \cite{BorkarConv} was used, e.g., in
\cite{YukLinZeroDelay} and \cite{ComYukTak09}, while in
\cite{BorkarMitterTatikonda}, \cite{PCRWIT08}, \cite{PWGIT09}, \cite{TakMit09},
\cite{AsnaniWeissman}, \cite{SPK1}, and \cite{SPK2} the \emph{average cost
  optimality equation} was used (typically through the \emph{vanishing discount
  method}). In particular, \cite{PCRWIT08}, \cite{SPK1}, and \cite{SPK2} use
this latter approach to solve dynamic programs that provide \emph{explicit}
channel capacity expressions.  In this paper (unlike in our earlier work
\cite{YukLinZeroDelay}), we also use the average cost optimality equation
approach, but here certain technical subtleties complicate the analysis: (i) the
structural result (on the optimality of Walrand-Varaiya type policies) only
holds for finite horizon problems; and (ii) we have a controlled Markov chain
(where the beliefs are the states and the quantizer maps are the actions) only
when the quantizers belong to the Walrand-Varaiya class (see
Definition~\ref{WVdef}). Much of our technical analysis concerns extending this
line of argument to the infinite horizon case through the study of recurrence,
coupling, convergence, and continuity properties of the underlying controlled
Markov chain.

\section{The Finite Horizon Average Distortion  and The Infinite Horizon
  Discounted Distortion Problems}

\label{chap:finite}

\subsection{The Finite Horizon Average Cost Problem}

In view of Theorem~\ref{WalrandVaraiyaTheorem}, for a finite horizon problem any
admissible (i.e., zero delay) quantization policy can be replaced by a
Walrand-Varaiya type policy. Using the terminology of Markov decision processes,
we will also refer to such policies as Markov policies. The class of all such
policies is denoted by $\Pi_W$, and is formally defined below.

\begin{defn} \label{WVdef} Let $\Q$ denote the set of all quantizers $Q:\sX\to
  \sM$. An (admissible) quantization policy $\Pi=\{\eta_t\}$ belongs to $\Pi_W$
  if there exists a sequence of mappings $\{\heta_t\}$ of the type $\heta_t:
  \P(\sX) \to \Q$ such that for $Q_t=\heta_t(\pi_t)$ we have
  $q_t=Q_t(X_t)=\eta_t(I_t)$. A policy in $\Pi_W$ is called \emph{stationary} if
  $\heta_t$ does not depend on $t$.  The set of such stationary policies is
  denoted by $\Pi_{WS}$.
\end{defn}

\noindent\emph{Remark}.   It is worth pointing out that the classical definition of a stationary (time
invariant or sliding block \cite{Gra11}) encoder involves a ``two sided''
infinite source sequence $\{X_t\}_{t=-\infty}^{\infty}$ and has the form
$q_t=g(X_{[-\infty,t]})$ for all $t$, where $g$ maps the infinite past
 $X_{[-\infty,t]}= \ldots,X_{t-2},X_{t-1},X_t,$ up to time $t$ into the symbol
$q_t$. Clearly, for a ``one sided'' source $\{X_t\}_{t\ge 0}$ such a definition
of stationary codes is problematic. Thus, in a sense, stationary
Walrand-Varaiya type encoding policies give a useful generalization of
classical stationary encoders  for the case of one sided sources.

\medskip

Building on \cite{YukIT2010arXiv} and \cite{YukLinZeroDelay}, suppose a given
quantizer policy $\Pi=\{\heta_t\}$ in $\Pi_W$ is used to encode the Markov
source $\{X_t\}$.  Let $P=P(x_{t+1}|x_t)$ denote the transition kernel of the
source.  Observe that the conditional probability of $q_t$ given $\pi_t$ and
$x_t$ is given by $P(q_t | \pi_t,x_t) = 1_{\{Q_t(x_t)=q_t\}}$ with
$Q_t=\heta_t(\pi_t)$, and is therefore determined by the quantizer policy. Then
standard properties of conditional probability can be used to obtain the
following ``filtering equation'' for the evolution of $\pi_t$:
\begin{eqnarray}
\pi_{t+1}(x_{t+1})\!\!\!\! & = & \!\!\!\!  \frac{\sum_{x_t} \pi_t(x_t) P(q_t | \pi_t, x_t)
  P(x_{t+1}|x_t)}{\sum_{x_t} \sum_{x_{t+1}} \pi_t(x_t) P(q_t| \pi_t, x_t) P(x_{t+1}|x_t)} \nonumber \\
&=&  \!\!\!\!   \frac{1}{\pi_t(Q^{-1}_t(q_t))} \sum_{x_t\in Q^{-1}_t(q_t)} P(x_{t+1}|
x_t)\pi_t(x_t).  \label{filtre}
\end{eqnarray}
Therefore, given $\pi_t$ and $Q_t$, $\pi_{t+1}$ is conditionally independent of
$(\pi_{[0,t-1]},Q_{[0,t-1]})$. Thus $\{\pi_t\}$ can be viewed as a
$\P(\sX)$-valued controlled Markov process \cite{HernandezLermaMCP} with
$\Q$-valued control $\{Q_t\}$ and average cost up to time $T-1$ given by
\[
E^{\Pi}_{\pi_0}\left[\frac{1}{T}   \sum_{t=0}^{T-1} d(X_t,\hX_t)\right] =  E^{\Pi}_{\pi_0}\left[ \frac{1}{T}\sum_{t=0}^{T-1} c(\pi_t,Q_t)\right],
\]
where
\begin{equation}
\label{eq_cdef}
c(\pi_t,Q_t): = \sum_{i=1}^M \min_{\hx\in \hat{\sX}} \sum_{x \in Q_t^{-1}(i)}  \pi_t(x) d(x,\hx).
\end{equation}
In this context, $\Pi_W$ corresponds to the class of deterministic Markov
control policies \cite{HernandezLermaMCP}. The Appendix provides a brief
overview of controlled Markov  processes. 

The following statements follow from results in \cite{YukLinZeroDelay}, but they
can also be straightforwardly derived since, in contrast to 
\cite{YukLinZeroDelay}, here we have only finitely many $M$-cell quantizers on
$\sX$. For any $\Pi\in \Pi_A$, define
\[
J_{\pi_0}(\Pi,T) \coloneqq E^{\Pi}_{\pi_0}\left[\frac{1}{T}   \sum_{t=0}^{T-1}
  d(X_t,\hX_t)\right]  .
\]

\begin{proposition}\label{MeasurableSelectionApplies}  For
  any $T\ge 1$, there exists a policy $\Pi$ in $\Pi_W$ such
  that
\begin{equation}
\label{eq_opt}
J_{\pi_0}(\Pi,T) = \inf_{\Pi' \in \Pi_A} J_{\pi_0}(\Pi',T).
\end{equation} 
Letting  $J^T_T(\,\cdot\,)\coloneqq0$, $J^T_0(\pi_0) \coloneqq \min_{\Pi\in \Pi_W}  J_{\pi_0}(\Pi,T)$,
the dynamic programming recursion
\begin{equation}
T J^T_t(\pi)  \nonumber = \min_{Q \in \mathcal{Q}} \bigg( c(\pi,Q) +  T E\bigl[J^T_{t+1}(\pi_{t+1})|\pi_t=\pi,Q_t=Q\bigr] \bigg)
\end{equation}
holds for all $t=T-1,T-2,\ldots,0$ and $\pi\in \P(\sX)$.
\end{proposition}

\begin{proof}
  By Theorem \ref{WalrandVaraiyaTheorem}, there exists a policy $\Pi$ in $\Pi_{W}$
  such that (\ref{eq_opt}) holds.  Also, by Theorem~\ref{thm:bellman} in the
  Appendix, we can use the dynamic programming recursion to solve for an optimal
  quantization policy $\Pi \in \Pi_W$.
\end{proof}

\subsection{The Infinite Horizon Discounted Cost Problem}

As discussed in Section~\ref{sec:problemdefinition}, the goal of the infinite horizon discounted cost problem is to find policies that achieve
\begin{equation}
\label{LQGopt22}
J^{\beta}_{\pi_0}\coloneqq\inf_{\Pi\in \Pi_A} J_{\pi_0}^{\beta}(\Pi)
\end{equation}
for given $\beta \in (0,1)$, where
\begin{equation*}
J^{\beta}_{\pi_0}(\Pi) = \lim_{T \to \infty} E^{\Pi}_{\pi_0}\left[\,
  \sum_{t=0}^{T-1} \beta^t d(X_t,\hX_t)\right]. 
\end{equation*}

From the viewpoint of source coding, the discounted cost problem has much less
significance than the average cost problem.  However the discounted cost
approach will play an important role in deriving results for the average cost
problem.  

\begin{proposition}
\label{thm:dc}
There exists an optimal (deterministic) quantization policy in $\Pi_{WS}$ among all
policies in $\Pi_A$ that achieves the infimum in~(\ref{LQGopt22}).
\end{proposition} 

\begin{proof} Observe that
  \begin{eqnarray}
\lefteqn{ \inf_{\Pi \in \Pi_A} \lim_{T \to \infty} E^{\Pi}_{\pi_0}\left[\,
  \sum_{t=0}^{T-1} \beta^t d(X_t,\hX_t)\right] }   \nonumber \qquad \qquad \qquad \\
 &\geq & \limsup_{T \to \infty}
\inf_{\Pi\in \Pi_A}  E^{\Pi}_{\pi_0}\left[\, \sum_{t=0}^{T-1} \beta^t
  d(X_t,\hX_t)\right]  \nonumber  \\ 
&= & \limsup_{T \to \infty} \min_{\Pi\in \Pi_W}  E^{\Pi}_{\pi_0}\left[\,
  \sum_{t=0}^{T-1} \beta^t d(X_t,\hX_t)\right]\nonumber \\
&= & \limsup_{T \to \infty} \min_{\Pi\in \Pi_W}  E^{\Pi}_{\pi_0}\left[\,
  \sum_{t=0}^{T-1} \beta^t c(\pi_t,Q_t)\right], \label{eqn:seqpiW}
\end{eqnarray}
where the first equality follows from Theorem~\ref{WalrandVaraiyaTheorem} and
the second from the definition of $c(\pi_t,Q_t)$ in \eqref{eq_cdef}. For each
$T$, let $\Pi_T$ denote the optimal policy in $\Pi_W$ achieving the  minimum in
(\ref{eqn:seqpiW}). 

One can easily check that conditions (i)--(iii) of Theorem~\ref{thm:iterative}
in the Appendix hold in our case (with $\sZ=\mathcal{P}(\sX)$, $\sA=\Q$,
$c(z,a)= c(\pi,Q)$, and $K(dz'|z,a)= P(d\pi'|\pi,Q)$). Specifically, the
definition of $c(\pi,Q)$ in \eqref{eq_cdef} shows that $c$ is continuous, so
(i) holds. Condition (ii) clearly holds. since $\Q$ is a finite set. Finally, it
is easily verified that the stochastic kernel $P(d\pi_{t+1} | \pi_t,Q_t)$ is
weakly continuous, i.e., that $\int_{\mathcal{P}(\sX)} f(\pi') P(d\pi'|\pi,Q)$ is
continuous on $\mathcal{P}(\sX) \times \Q$ for any continuous and bounded $f:
\mathcal{P}(\sX)\to \R$ (see \cite[Lemma 11]{YukLinZeroDelay}). Thus by
Theorem~\ref{thm:iterative} in the Appendix, this sequence of policies,
$\{\Pi_T\}$, can be obtained by using the iteration algorithm
\begin{equation*}
J_t(\pi) = \min_{Q \in \Q} \left[ c(\pi, Q) + \beta \int_{\mathcal{P}(\sX)}
  J_{t-1}(\pi') P(d\pi'|\pi,Q) \right] 
\end{equation*}
with $J_0(\pi) \equiv 0$.  By the same theorem, the sequence of value functions
for the policies $\{\Pi_T\}$, i.e.\ $\{J_{\pi_0}(\Pi_W,T)\}$, converges to the
value function of some deterministic policy $\Pi \in \Pi_{WS}$ (i.e., a
deterministic stationary Markov policy) which is optimal in the set of policies
$\Pi_W$ for the infinite horizon discounted cost problem. Thus by the chain of
inequalities leading to (\ref{eqn:seqpiW}), $\Pi$ is optimal among all policies
in $\Pi_A$.
\end{proof}

\section{Main Results: The infinite Horizon Average Distortion Problem} 
\label{chap:avgcost}

The more challenging average cost case deals with 
a performance measure (the long time average distortion) that is usually
studied in source coding problems. Formally, the infinite horizon average cost of
a coding policy $\Pi$ is
\begin{equation}
\label{infiniteCost2}
J_{\pi_0}(\Pi)  = \limsup_{T \to \infty}
E^{\Pi}_{\pi_0}\left[\frac{1}{T}   \sum_{t=0}^{T-1} d(X_t,\hX_t)\right]
\end{equation}
and the goal is to find an optimal policy attaining
\begin{eqnarray}\label{LQGoptAVG}
J_{\pi_0} \coloneqq   \inf_{\Pi\in \Pi_A} J_{\pi_0}(\Pi) .
\end{eqnarray}

\subsection{Optimality of policies in $\Pi_W$ for stationary sources}

For the infinite horizon setting structural results such as 
Theorems~\ref{witsenhausenTheorem} and \ref{WalrandVaraiyaTheorem} are not
available in the literature as the proofs are based on dynamic programming,
which starts at a finite terminal time stage and optimal policies are computed
by working backwards from the end. However, as in \cite{YukLinZeroDelay}, we can
prove an infinite horizon analog of Theorem~\ref{WalrandVaraiyaTheorem} assuming
that  $\{X_t\}$  starts from its  invariant measure $\pi^*$ (which exists e.g.\ if
 $\{X_t\}$ is irreducible and aperiodic). 

\begin{proposition}[\protect{\cite[Theorem~6]{YukLinZeroDelay}}]
\label{thm_QWopt}
Assume $\{X_t\}$ is a stationary Markov chain with invariant probability
$\pi^*$. Then there exists an optimal policy in $\Pi_W$ that solves the
minimization problem (\ref{LQGoptAVG}), i.e., there exists $\Pi\in \Pi_W$ such
that
\[
J_{\pi^*}(\Pi) = J_{\pi^*}.
\]
\end{proposition}

The proof of the proposition is straightforward; it relies on a
construction that pieces together policies from $\Pi_W$ that on time segments of
appropriately large lengths increasingly well approximate the infimum of the
infinite horizon cost achievable by policies in $\Pi_A$; see
\cite{YukLinZeroDelay} for the details. This construction results in a policy
that is nonstationary in general.  However, for the finite alphabet case
considered here, we will also establish the optimality of deterministic
\emph{stationary} policies even for possibly nonstationary Markov sources. The
remainder of the section focuses on this problem.

\subsection{Optimality of Stationary Coding Policies}
\label{sec:nonstat}

The following theorem is the main result of the paper. It states that for any
irreducible and aperiodic Markov source there exists a stationary Markov
(Walrand-Varaiya type) coding policy that is optimal among all zero delay coding
policies. Note that the theorem does not require the source to be stationary.

\begin{thm}\label{mainThm}
Assume that $\{X_t\}$ is an irreducible and aperiodic Markov chain. Then for any
initial distribution $\pi_0$, 
\[
\inf_{\Pi \in \Pi_A}  J_{\pi_0}(\Pi)  = \min_{\Pi \in \Pi_{WS}} J_{\pi_0}(\Pi).
\]
Furthermore, there exist $\Pi^*\in \Pi_{WS}$ that achieves the minimum above
simultaneously for all $\pi_0$ and  which satisfies for all $T\ge 1$
\begin{equation}
  \label{eq:convrate}
  \frac{1}{T}
E_{\pi_0}^{\Pi^*} \left[ \sum_{t=0}^{T-1} d(X_t,\hX_t) \right] \le J_{\pi_0} +
\frac{K}{T}
\end{equation}
for some positive constant $K$. 
\end{thm}

The theorem is proved in the next subsection where the constant $K$ is more
explicitly identified. Here we give a brief description of the main steps.  The
key step in the proof is Lemma~\ref{sim_arg_lemma} where we  build on the approach
of Borkar~\cite{borkar2000average} (but use a different construction) to show
that for any two initial distributions $\mu_0$ and $\zeta_0$, the absolute
difference of the optimal infinite horizon discounted costs $J^{\beta}_{\mu_0}$
and $J^{\beta}_{\zeta_0} $ is uniformly upper bounded by a constant times the
$L_1$ Wasserstein distance between $\mu_0$ and $\zeta_0$. With the aid of this
result and an Abelian lemma that relates the infinite horizon discounted cost to
the average cost, Lemma~\ref{mainThm1} shows through the vanishing discount
approach that for the infinite horizon average cost problem, \emph{randomized}
stationary Markov policies are at least as good as deterministic policies in
$\Pi_A$. Lemma~\ref{mainThm2} in turn shows that deterministic stationary Markov
policies are no worse than randomized ones, which, together with
Lemma~\ref{mainThm1}, gives
$\inf_{\Pi \in \Pi_A} J_{\pi_0}(\Pi) = \inf_{\Pi \in \Pi_{WS}} J_{\pi_0}(\Pi)$
(Lemma~\ref{keyConnection}). Finally, we show that Lemma~\ref{sim_arg_lemma}
implies that the \emph{average cost optimality equation} (ACOE) (see
Theorem~\ref{thm3} in the Appendix) holds for our controlled Markov chain, which
in turn implies that the infimum $\inf_{\Pi \in \Pi_{WS}} J_{\pi_0}(\Pi)$ is
achieved by some policy in $\Pi_{WS}$, proving the first statement of the
theorem. The $O(1/T)$ convergence rate result is shown to be a direct
consequence of the ACOE.

\begin{defn}[$\epsilon$-Optimality]
  Given an initial distribution $\pi_0$ and $\epsilon>0$, a policy $\Pi\in \Pi_A$
  is \emph{$\epsilon$-optimal} if $J_{\pi_0}(\Pi) \leq J_{\pi_0} + \epsilon$,
  where $J_{\pi_0}$ is the optimal performance for the infinite horizon average
  cost problem.
\end{defn}

Now suppose that $\{X_t\}$ is irreducible and aperiodic and it starts from the
unique invariant probability $\pi^*$ so that it is a stationary 
process. Consider the (nonstationary) coding policy that is obtained by
periodically extending an initial segment of the optimal stationary policy
$\Pi^*$ in Theorem~\ref{mainThm}. In particular, assume $\Pi^*=\{\eta^*\}$ and
for $T\ge 1$  consider the periodic policy $\Pi^{(T)}=\{\eta^{(T)}_t\}$, where
$\eta^{(T)}_t=\eta^*$ for $t=kT+1,\ldots,(k+1)T$, $k=0,1,2,\ldots,$ and
$\eta^{(T)}_t\equiv \eta^*(\pi^*)$ for $t=kT$, $k=0,1,2,\ldots$. Since
$\{X_t\}$ is stationary, the infinite horizon cost of $\Pi^{(T)}$ is
\[
  J_{\pi^*}(\Pi^{(T)}) =   \frac{1}{T}
E_{\pi_0}^{\Pi^*} \left[ \sum_{t=0}^{T-1} d(X_t,\hX_t) \right].
\]
Since the encoder of $\Pi^{(T)}$ is reset to $\eta^*(\pi^*)$ each time after
processing $T$ source samples, we can say that it has memory length $T$. The
following result, which is implied by the above construction and the bound
\eqref{eq:convrate}, may have implications in the construction of practical
codes since, loosely speaking, the complexity of a code is determined by its
memory length.

\begin{thm}\label{epsOptimalityTheorem}
  Assume $\{X_t\}$ is an irreducible and aperiodic Markov chain. If $X_0 \sim
  \pi^*$, where $\pi^*$ is the invariant probability measure, then for every
  $\epsilon>0$, there exists a finite memory, nonstationary, but periodic
  coding policy with period at most $\frac{K}{\epsilon}$ that is 
  $\epsilon$-optimal, where $K$ is 
  the constant from Theorem~\ref{mainThm}.
\end{thm}

\subsection{Proof of Theorem~\ref{mainThm}}

Let $\sX = \{1,\cdots,|\sX|\}$ be viewed as a subset of $\mathbb{R}$.  The
$L_1$ \emph{Wasserstein distance} \cite{villani2008optimal} between two
distributions $\mu_0$ and $\zeta_0$ is defined as
\begin{equation}
\label{def:wass}
\rho_1(\mu_0,\zeta_0) \coloneqq   \inf_{X\sim \mu_0, Y\sim \zeta_0}  E\big[|X-Y|\big], 
\end{equation}
where the infimum is taken over all joint distributions of pairs of $\sX$-valued
random variables $(X,Y)$ such that $X\sim \mu_0$ and $Y\sim \zeta_0$. It can be
shown that the infimum in the definition is in fact a minimum and that the $L_1$
Wasserstein distance is a metric on $\mathcal{P}(\sX)$. 

Recall the definition
\[
J^{\beta}_{\pi_0}\coloneqq  \inf_{\Pi\in \Pi_A} \lim_{T \to \infty} E^{\Pi}_{\pi_0}\left[\,
  \sum_{t=0}^{T-1} \beta^t d(X_t,\hX_t)\right] .
\]

The following lemma is a key step in the proof.

\begin{lem}
\label{sim_arg_lemma}
Suppose the source is an irreducible and aperiodic Markov chain.  Then for any
pair of initial distributions $\mu_0$ and $\zeta_0$, and any $\beta \in (0,1)$,
we have
\[
\big|J^{\beta}_{\mu_0} -J^{\beta}_{\zeta_0} \big| \le   K_1 \|d\|_{\infty}
\rho_1(\mu_0,\zeta_0), 
\]
where  $K_1$  is a finite constant and $\|d\|_{\infty} = \max_{x,y} d(x,y)$. 
\end{lem}

\begin{proof}
Note that by  monotone convergence  for any $\Pi$ and $\beta\in (0,1)$, 
\[
\lim_{T\to \infty} E\biggl[\sum_{t=0}^{T-1} \beta^t   d(X_t,\hX_t)\biggr] =
E\biggl[ \sum_{t=0}^{\infty} \beta^t   d(X_t,\hX_t) \biggr]. 
\]
Thus the lemma statement is equivalent to
\[
\bigg|\inf_{\Pi \in \Pi_A}E^{\Pi}_{\mu_0}\bigg[\sum_{t=0}^{\infty} \beta^t
  d(X_t,\hX_t)\bigg] - \inf_{\Pi \in \Pi_A} E^{\Pi}_{\zeta_0}\bigg[\sum_{t=0}^{\infty} \beta^t
  d(X_t,\hX_t)\bigg] \bigg| \leq K_1 \|d\|_{\infty} \rho_1(\mu_0,\zeta_0) .
\]

The proof builds on the approach of Borkar~\cite{borkar2000average} (see also
\cite{borkar2007dynamic} and \cite{BorkarMitterTatikonda}), but our argument 
is different (and also more direct) since the absolute continuity conditions in
\cite{borkar2000average} are not applicable here due to quantization.
 As in \cite{BorkarMitterTatikonda}, in the proof we will enlarge
the space of admissible coding policies to allow for randomization at the
encoder.  Since for a discounted infinite horizon optimal encoding problem
optimal policies are deterministic even among possibly randomized policies (see
Proposition~\ref{thm:dc}), allowing common randomness does not change the
optimal performance.

In our construction, we will use the well known coupling method for Markov
chains. In particular, we will use a minor modification of the coupling argument
in the proof of \cite[Theorem~1.8.3]{Nor97}. 

Given the initial distributions $\mu_0$ and $\zeta_0$, let $\lambda$ be  a
joint distribution on $\sX\times \sX$ , having marginals $\mu_0$ and $\zeta_0$,
such that $E\big[|X-Y|\big]=\rho_1(\mu_0,\zeta_0)$ if $(X,Y)\sim \lambda$.  

Now consider the $\sX\times \sX$-valued process $\{(X_t,Y_t)\}_{t\ge 0}$ such
that $\{X_t\}_{t\ge 0}\sim (\mu_0,P)$, $\{Y_t\}_{t\ge 0}\sim (\zeta_0,P)$,
$(X_0,Y_0)\sim \lambda$, and $\{X_t\}_{t\ge 1} $ and $\{Y_t\}_{t\ge 1} $ are
conditionally independent given $(X_0,Y_0)$. We note that given $\{X_t\}\sim
(\mu_0,P)$, a process $\{(X_t,Y_t)\}$ with such a distribution can be obtained
via an i.i.d.\ randomization process $\{W_t\}$ which is uniform on the interval
$[0,1]$ and is independent of $\{X_t\}$, and via  appropriate functions $F_0,F:
\sX\times [0,1]\to \sX$, by letting
\begin{equation}
  \label{eq:ygen}
Y_0=F_0(X_0,W_0) \text{ and } Y_t=F(Y_{t-1},W_t) \quad  \text{for all $t\ge 1$.}
\end{equation}

Fixing a reference state $b\in \sX$,
define
\[
\tau =\inf \{t\ge 0: X_n=Y_n=b\}.
\]
Since the common transition probability $P$ of $\{X_t\}$ and $\{Y_t\}$ is
irreducible and aperiodic, it easily follows that $\{(X_t,Y_t)\}$ is an
irreducible and aperiodic Markov chain \cite[p.~41]{Nor97}. Since $\sX$ is
finite, this implies that the chain is positive recurrent and thus
 $E[\tau]<\infty$. Define $X'_t=X_t$ for  $t\geq 0$
so that $\{X'_t\}\sim (\mu_0,P)$.  Also define the process $\{X''_t\}$ by
\[
  X''_t=
  \begin{cases}
    Y_t & \text{ if $t\le \tau$}\\
     X_t   & \text{ if $t > \tau$}.
  \end{cases}
\]
It is shown in \cite[p.~42]{Nor97} that  $\{X''_t\}$  is a Markov chain such that
$ \{X''_t\}\sim (\zeta_0,P)$. 

Assume without loss of generality that $J^\beta_{\mu_0} - J^\beta_{\zeta_0}\ge
0$. Then from the above
\begin{eqnarray}
\big| J^\beta_{\mu_0} - J^\beta_{\zeta_0}\big| & =&  J^\beta_{\mu_0} -
J^\beta_{\zeta_0} \nonumber \\*
&=& \inf_{\Pi\in \Pi_A}E^{\Pi}_{\mu_0}\bigg[\sum_{t=0}^{\infty} \beta^t
d(X'_t,\hX'_t)\bigg] - \inf_{\Pi\in \Pi_A} E^{\Pi}_{\zeta_0}\bigg[\sum_{t=0}^{\infty} \beta^t
d(X''_t,\hX''_t)\bigg]   \label{eq:muzetadiff} \\
&=& E^{\Pi'}_{\mu_0}\bigg[\sum_{t=0}^{\infty} \beta^t
d(X'_t,\hX'_t)\bigg] -  E^{\Pi''}_{\zeta_0}\bigg[\sum_{t=0}^{\infty} \beta^t
d(X''_t,\hX''_t)\bigg],  \label{eq:muzetadiff1}
\end{eqnarray}
where $\Pi'$  (resp.\ $\Pi''$) achieves the first (resp.\
the second) infimum  in \eqref{eq:muzetadiff}; see Proposition~\ref{thm:dc}. 

Consider the following suboptimal coding and decoding policy for $\{X'_t\}$: In
addition to observing the source $X'_t=X_t$, $t\ge 0$, the encoder is
also given the randomization process $\{W_t\}$ which is independent of
$\{X'_t\}$. Then the encoder can generate $Y_0,\ldots,Y_{\tau}$ according
to the representation \eqref{eq:ygen} and thus it can produce the second source
process $\{X''_t\}$. The encoder for $\{X'_t\}$ feeds sequentially the
$\{X''_t\}$ values to the quantizer policy  $\Pi''$ and produces the same
channel symbols $q''_t$ and reproduction sequence
$\hX''_t=\gamma''_t(q''_{[0,t]})$ as  the policy $\Pi''$ does in response to
$\{X''_t\}$.  Note that this procedure comprises a 
suboptimal {\it randomized encoder} and a {\it deterministic decoder} for coding
$\{X'_t\}$. Let us denote this randomized policy by $\hat{\Pi}$. Thus we obtain
the upper bound  
\[
E^{\Pi'}_{\mu_0}\bigg[\sum_{t=0}^{\infty} \beta^t
  d(X'_t,\hX'_t)\bigg] \le E^{\hat{\Pi}}_{\mu_0}\bigg[\sum_{t=0}^{\infty} \beta^t
  d(X'_t,\hX''_t)\bigg].
\]
In view of this and  \eqref{eq:muzetadiff1},  we can write 
\begin{eqnarray}\label{bound1}
\lefteqn{\big| J^\beta_{\mu_0} - J^\beta_{\zeta_0}\big| } \nonumber \\
&\le &   E^{\hat{\Pi}}_{\mu_0}\bigg[\sum_{t=0}^{\infty} \beta^t
  d(X'_t,\hX''_t)\bigg] - E^{\Pi''}_{\zeta_0}\bigg[\sum_{t=0}^{\infty} \beta^t
  d(X''_t,\hX''_t)\bigg] \\
& \leq & \bigg | E \bigg[\sum_{t=0}^{\infty} \beta^t \big( d(X'_t,\hX''_t) -
d(X''_t,\hX''_t) \big)\bigg] \bigg| \nonumber \\ 
& \leq &  E[\tau] \|d\|_{\infty},  \label{eq:cnullbound}
\end{eqnarray}
where the last inequality follows  since $X'_t=X''_t$ if $t\ge 
\tau$. 

On the other hand, 
\[
E[\tau]= \sum_{x,y} \lambda(x,y) E[\tau|X_0=x,Y_0=y] 
\]
and since $E[\tau]<\infty$, we have that $K_1\coloneqq \max_{x,y}
E[\tau|X_0=x,Y_0=y]<\infty$ and 
\begin{equation*}
E[\tau] \leq \sum_{x\neq y}\lambda(x,y) K_1 = \Pr(X_0 \neq Y_0) K_1 \leq K_1 \rho_1(\mu_0,\zeta_0),
\end{equation*}
where the the second inequality  follows
from the fact that $\Pr(X_0 \neq Y_0) \leq \rho_1(\mu_0,\zeta_0)$ by \eqref{def:wass}.  This and
\eqref{eq:cnullbound} complete the proof of  Lemma~\ref{sim_arg_lemma}.  
\end{proof}

Under any given stationary Markov policy $\Pi\in \Pi_{WS}$ the sequence
$\{(\pi_t,Q_t)\}_{t\ge 0}$ is a $\mathcal{P}(\sX)\times \Q$-valued Markov chain
whose transition kernel is determined by $\Pi$ and the transition kernel
$P(d\pi_{t+1} | \pi_t,Q_t)$, which is given by the filtering equation
\eqref{filtre} and does not depend on $\Pi$.  As pointed out in the proof of
Proposition~\ref{thm:dc}, the transition kernel $P(d\pi_{t+1} | \pi_t,Q_t)$ is
weakly continuous. This implies that the Markov process $\{(\pi_t,Q_t)\}$ is
weak Feller, that is, the transition kernel $P(d(\pi_{t+1},Q_{t+1})|\pi_t,Q_t)$
is weakly continuous \cite[C.3~Definition]{HernandezLermaMCP}. Since every weak
Feller Markov process with a compact state space has an invariant probability
measure \cite{YukMeynTAC2010}, it follows that there exists a probability
measure $\pi^*(\Pi)$ on $\mathcal{P}(\sX)$ such that if $\pi_0$ is picked
randomly according to $\pi^*(\Pi)$, then $\{(\pi_t,Q_t)\}$ is a stationary
process. We call $\pi^*(\Pi)$ an invariant probability on ${\cal P}(\sX)$
induced by $\Pi \in \Pi_{WS}$.

Note that if the initial probability $\pi_0$ is random with distribution
$\pi^*(\Pi)$, the quantization policy $\Pi$ becomes a randomized policy since
the encoder and decoder must have access to the same random
$\pi_0$. Expectations under such policies will be denoted by $E^{\Pi}_{\pi_0\sim
  \pi^*(\Pi)}$. 

\begin{lem}\label{mainThm1}
  If the source is irreducible and aperiodic, then for any  initial distribution
  $\pi$, 
\begin{eqnarray*}
&&\inf_{\Pi \in \Pi_A} \limsup_{T\to\infty} \frac{1}{T} E_{\pi}^\Pi \left[
  \sum_{t=0}^{T-1} d(X_t,\hX_t) \right] \nonumber \\*
&& \quad \quad \geq \inf_{\Pi \in \Pi_{WS}} \limsup_{T\to\infty} \frac{1}{T}
E_{\pi_0\sim\pi^*(\Pi)}^\Pi \left[ \sum_{t=0}^{T-1} d(X_t,\hX_t).
\right] 
\end{eqnarray*}
\end{lem}

\begin{proof}
We will need the following well known Abelian result.

\begin{lem}[\protect{\cite[Lemma~5.3.1]{HernandezLermaMCP}}]
\label{lem:abel}
Let $\{c_t\}_{t \geq 0}$ be a sequence of nonnegative numbers.  Then
\begin{align*}
\liminf_{T\to\infty} \frac{1}{T} \sum_{t=0}^{T-1} c_t &\leq \liminf_{\beta \uparrow 1} (1-\beta) \sum_{t=0}^\infty \beta^t c_t \\
&\leq \limsup_{\beta \uparrow 1} (1-\beta) \sum_{t=0}^\infty \beta^t c_t \\
&\leq \limsup_{T\to\infty} \frac{1}{T} \sum_{t=0}^{T-1} c_t.
\end{align*}
\end{lem}

Let $\{\Pi_k\}$ be a sequence of policies in $\Pi_A$ such that $
\lim_{k\to\infty} J_{\pi}(\Pi_k) = J_{\pi}$ and fix $n>0$ such that
\begin{equation}
\label{eqn:firstepsilon}
J_{\pi} \geq J_{\pi}(\Pi_n) - \epsilon.
\end{equation}
Applying  Lemma \ref{lem:abel} with $c_t=E^{\Pi_n}_{\pi}\bigl[d(X_t,\hat{X}_t)\bigr]$, there exists $\beta_\epsilon\in (0,1)$ such that for
all $\beta\in (\beta_{\epsilon},1)$
\begin{align}
J_{\pi} \geq J_{\pi}(\Pi_n) - \epsilon &\geq (1-\beta) E_{\pi}^{\Pi_n} \left[
  \sum_{t=0}^{\infty} \beta^t d(X_t,\hX_t) \right] -
2\epsilon  \nonumber  \\ 
&\geq \min_{\Pi \in \Pi_{WS}} (1-\beta) E_{\pi}^{\Pi} \left[
  \sum_{t=0}^{\infty} \beta^t d(X_t,\hX_t) \right] - 2\epsilon, \label{eqn:epsilon2} 
\end{align}
where the minimum exists by Proposition~\ref{thm:dc}.

Now consider the case where for some $\Pi\in \Pi_{WS}$ the initial measure
$\pi_0$ is distributed according to $\pi_0 \sim \pi^*(\Pi)$.  For ease of
interpretation, let $X'_t$ denote the source process with $X'_0 \sim \pi$, let
$X''_t$ be a process with $X''_0 \sim \pi_0$ for some fixed $\pi_0$, and
in addition let $X'_t$ and $X''_t$ be coupled as in Lemma~\ref{sim_arg_lemma}.
Then for any $\beta \in (0,1)$,
\begin{align}
\bigg| \min_{\Pi\in\Pi_{WS}} (1-\beta) E^{\Pi}_{\pi}\bigg[\sum_{t=0}^{\infty} \beta^t
  & d(X'_t,\hX'_t)\bigg] - \min_{\Pi\in\Pi_{WS}} (1-\beta) E^{\Pi}_{\pi_0}\bigg[\sum_{t=0}^{\infty} \beta^t
  d(X''_t,\hX''_t)\bigg] \bigg| \nonumber \\
& \leq  (1-\beta) E[\tau] \|d\|_{\infty} \nonumber \\
& \le (1-\beta)K_1  \|d\|_{\infty} \rho_1(\pi,\pi_0),   \label{eqn:2dists1} 
\end{align}
where $\tau = \min\{t \ge  0: X'_t =X''_t\}$, and where the  first inequality
follows from the coupling of the Markov 
chains as in Lemma~\ref{sim_arg_lemma} (see \eqref{eq:cnullbound}) and the
second also follows from the proof of Lemma~\ref{sim_arg_lemma}. 
Since $\rho_1(\pi,\pi_0)$ is  upper bounded by $|\sX|$ (see
\eqref{def:wass}) for any $\pi$ and $\pi_0$, we obtain 
\begin{align}
\bigg| \min_{\Pi\in\Pi_{WS}} (1-\beta) E^{\Pi}_{\pi}\bigg[\sum_{t=0}^{\infty} \beta^t
  & d(X'_t,\hX'_t)\bigg] - \inf_{\Pi\in\Pi_{WS}} (1-\beta)
  E^{\Pi}_{\pi_0\sim\pi^*(\Pi)}\bigg[\sum_{t=0}^{\infty} \beta^t 
  d(X''_t,\hX''_t)\bigg] \bigg| \nonumber \\
& \leq (1-\beta) K_2 \|d\|_{\infty}, \label{eqn:2dists} 
\end{align}
where $K_2=K_1|\sX|$. 

Choosing $\bar{\beta}$ such that \eqref{eqn:2dists} with $\beta=\bar{\beta}$ is
less than $\epsilon$, and combining the preceding bound with
\eqref{eqn:epsilon2} yields  for any $\beta\in
(\max\{\beta_{\epsilon},\bar{\beta}\}, 1)$,
\begin{align}
J_{\pi} &\geq \min_{\Pi\in\Pi_{WS}} (1-\beta) E_{\pi}^{\Pi} \left[
  \sum_{t=0}^{\infty} \beta^t d(X_t,\hX_t) \right] - 2\epsilon   \nonumber \\*
&\geq  \inf_{\Pi\in\Pi_{WS}}  (1-\beta)  E_{\pi_0\sim \pi^*(\Pi)}^{\Pi} \left[
  \sum_{t=0}^{\infty} \beta^t d(X_t,\hX_t) \right]  -
3\epsilon  \label{eq:inf-min-beta} \\ 
&\geq  (1-\beta)  E_{\pi_0\sim \pi^*(\Pi_{\beta})}^{\Pi_{\beta}} \left[
  \sum_{t=0}^{\infty} \beta^t d(X_t,\hX_t) \right]  - 4\epsilon  \nonumber \\ 
&\ge  \liminf_{T\to\infty} \frac{1}{T} E_{\pi_0\sim \pi^*(\Pi_{\beta})}^{\Pi_{\beta}} \left[
  \sum_{t=0}^{T-1} d(X_t,\hX_t) \right] - 5\epsilon   \nonumber \\ 
& = \limsup_{T\to\infty} \frac{1}{T} E_{\pi_0\sim \pi^*(\Pi_{\beta})}^{\Pi_{\beta}}  \left[
  \sum_{t=0}^{T-1} c(\pi_t,Q_t) \right] - 5\epsilon,   \nonumber
\end{align}
where the $\Pi_{\beta}\in \Pi_{WS}$ is chosen so that it achieves the infimum in
\eqref{eq:inf-min-beta} within $\epsilon$, and where the fourth inequality holds
by Lemma~\ref{lem:abel} if $\beta\in (\max\{\beta_{\epsilon},\bar{\beta}\}, 1)$
is large enough. Finally, the last equality follows since 
$\pi^*(\Pi_{\beta})$ is invariant and hence $\{(\pi_t, Q_t)\}$ is a stationary
process. Thus we obtain
\begin{align*}
J_{\pi} & \geq \limsup_{T\to\infty} \frac{1}{T} E_{\pi_0\sim \pi^*(\Pi_{\beta})}^{\Pi'}
\left[ \sum_{t=0}^{T-1} d(X_t,\hX_t) \right] - 5\epsilon \\ 
& \geq \inf_{\Pi\in\Pi_{WS}} \limsup_{T\to\infty} \frac{1}{T} E_{\pi_0\sim
  \pi^*(\Pi)}^{\Pi} \left[ \sum_{t=0}^{T-1} d(X_t,\hX_t) \right] - 5\epsilon, 
\end{align*}
where $\epsilon >0$ is arbitrary, which completes the proof. 
\end{proof}

\begin{lem}\label{mainThm2}
If the source is irreducible and aperiodic, then for any initial
distribution $\pi$, 
\begin{eqnarray*}
&&\inf_{\Pi \in \Pi_{WS}} \limsup_{T\to\infty} \frac{1}{T} E_{\pi}^\Pi \left[ \sum_{t=0}^{T-1} d(X_t,\hX_t) \right] \\*
&& \quad \quad = \inf_{\Pi \in \Pi_{WS}} \limsup_{T\to\infty} \frac{1}{T}
E_{\pi_0\sim \pi^*(\Pi)}^\Pi \left[ \sum_{t=0}^{T-1} d(X_t,\hX_t) \right].
\end{eqnarray*}
\end{lem}

\begin{proof}
First note that by  Lemma \ref{mainThm1},
\begin{eqnarray*}
&&\inf_{\Pi \in \Pi_{WS}} \limsup_{T\to\infty} \frac{1}{T} E_{\pi}^\Pi \left[ \sum_{t=0}^{T-1} d(X_t,\hX_t) \right] \\*
&& \quad \quad \geq \inf_{\Pi \in \Pi_{WS}} \limsup_{T\to\infty} \frac{1}{T}
E_{\pi_0\sim \pi^*(\Pi)}^\Pi \left[ \sum_{t=0}^{T-1} d(X_t,\hX_t) \right].
\end{eqnarray*}

Now apply the argument that led to the bounds \eqref{eqn:2dists1} and
\eqref{eqn:2dists}   to obtain 
\begin{eqnarray*}
\lefteqn{ \Bigg| \inf_{\Pi \in \Pi_{WS}} \limsup_{T\to\infty} \frac{1}{T} E_{\pi}^\Pi
\left[ \sum_{t=0}^{T-1} d(X_t,\hX_t) \right]  } \qquad\qquad \\*
 & & \mbox{}- \inf_{\Pi \in \Pi_{WS}}
\limsup_{T\to\infty} \frac{1}{T} E_{\pi_0\sim \pi^*(\Pi)}^\Pi \left[
  \sum_{t=0}^{T-1} d(X_t,\hX_t) \right] \Bigg| \hspace{0cm} \\
&  \leq & \limsup_{T\to\infty} \frac{1}{T}  E_{\pi_0\sim \pi^*(\Pi)}[\tau]
\|d\|_{\infty}  \\
&\le & \limsup_{T\to\infty} \frac{1}{T} K_2\|d\|_{\infty}  =0.
\end{eqnarray*}
\end{proof}

The following important  result immediately follows from Lemmas \ref{mainThm1} and
\ref{mainThm2}. 
\begin{lem}\label{keyConnection}
If the source is irreducible and aperiodic, then for any initial distribution $\pi_0$,
\begin{eqnarray}
&&\inf_{\Pi \in \Pi_A} \limsup_{T\to\infty} \frac{1}{T} E_{\pi_0}^\Pi \left[ \sum_{t=0}^{T-1} d(X_t,\hX_t) \right] \nonumber \\*
&& \quad \quad = \inf_{\Pi \in \Pi_{WS}} \limsup_{T\to\infty} \frac{1}{T} E_{\pi_0}^\Pi \left[ \sum_{t=0}^{T-1} d(X_t,\hX_t) \right]. \label{equalInfima}
\end{eqnarray}
\end{lem}

\noindent\emph{Remark.} This lemma is crucial because it shows that without any loss
we can restrict the search for optimal quantization policies to the set
$\Pi_{W}$. Since the filtering equation \eqref{filtre} leads to a controlled
Markov chain only for policies in $\Pi_W$, this lemma allows us to apply
controlled Markov chain techniques in the study of the the average distortion
problems. The rigorous justification of this fact is one of the main
contributions of this paper.

\medskip

Note that Lemma~\ref{keyConnection} immediately implies the first statement of
Theorem~\ref{mainThm} once we can show that the infimum in (\ref{equalInfima})
is actually a minimum. This will be done by invoking the ACOE for controlled
Markov chains.  To show that the infimum is achieved by a stationary and
deterministic Markov policy $\Pi \in \Pi_{WS}$ we make use of Theorem~\ref{thm3}
in the Appendix.  To do this we have to verify that the conditions of the
theorem are satisfied with $\sZ=\mathcal{P}(\sX)$, $\sA=\Q$, $c(z,a)= c(\pi,Q)$,
and $K(dz'|z,a)= P(\pi'|\pi,Q)$. We have already shown in the proof of
Proposition~\ref{thm:dc} that conditions (i)---(iii) hold.  Since
$\mathcal{P}(\sX)$ is the standard probability simplex in $\R^{|\sX|}$ and $\Q$
is a finite set, condition (iv) clearly holds. Finally, condition (v) holds
since the family of functions
\[
\big\{h_{\beta}(\zeta)\coloneqq   J^{\beta}_{\zeta} - J^{\beta}_{\zeta_0}: \beta\in (0,1)\big\}
\]
for some arbitrary but fixed $\zeta_0\in \mathcal{P}(\sX) $  is equicontinuous
on $\mathcal{P}(\sX)$  by  Lemma~\ref{sim_arg_lemma} which states that 
\begin{equation}
\label{eq:unifbound}
|h_{\beta}(\zeta)-h_{\beta}(\zeta')|= \big|J^{\beta}_{\zeta} - J^{\beta}_{\zeta'}\big|  \leq
K_1 \|d\|_{\infty} 
\rho_1(\zeta,\zeta').
\end{equation} 
Thus we can apply Theorem~\ref{thm3} to deduce the existence of a policy in $\Pi_{WS}$ achieving the
minimum in \eqref{equalInfima}. This completes the proof of the first statement in
Theorem~\ref{mainThm}. 

To prove the second statement \eqref{eq:convrate} we use the result in
\eqref{eq:convrate-triplet} in the Appendix. Note that by
Lemma~\ref{sim_arg_lemma} we have for all $\beta\in (0,1)$ and $\zeta\in
\mathcal{P}(\sX)$
\[
|h_{\beta}(\zeta)|= \big|J^{\beta}_{\zeta} - J^{\beta}_{\zeta'}\big|  \leq
\frac{K}{2} , 
\]
where 
\[
K\coloneqq 2K_1 \|d\|_{\infty} 
\rho_1(\zeta,\zeta_0) \le 2 K_1 \|d\|_{\infty} |\sX|.
\]
Thus equation \eqref{eq:convrate-triplet} implies,  with $z_0=\pi_0$,  $g^*=J_{\pi_0}$,
and $\Pi^*$ being the  optimal policy in $\Pi_{WS}$ achieving the minimum
in \eqref{equalInfima}, that 
\[
J_{\pi_0}(\Pi^*,T) - J_{\pi_0} \le \frac{K}{T}
\]
as claimed. \hfill \IEEEQED

\section{Zero-Delay Coding over a Noisy Channel with Feedback}

\label{secnoisy}

In this section, we briefly describe  the extension of our main results to zero 
delay lossy coding over  a noisy
channel. As in Section~\ref{sec:problemdefinition}, the encoder processes the
observed information source without delay. It is assumed that the 
source $\{X_t\}_{t\ge 0}$ is a discrete time Markov process with finite alphabet 
$\sX$. The encoder encodes  the source samples without delay  and
transmits the encoded versions to a receiver over a discrete channel with input
alphabet $ \sM=\{1,\ldots,M\}$ and output alphabet $\sM'\coloneqq\{1,\ldots,M'\}$, where $M$ and $M'$ are positive integers.

In contrast with the setup described in Section~\ref{sec:problemdefinition},
here the channel between the encoder and decoder is a discrete and memoryless \emph{noisy}
channel characterized by the transition probability $T(b|a)=\Pr(q'=b|q=a)$,
$a\in \sM$, $b\in \sM'$.

We assume that the encoder has access to the previous channel outputs in the
form of \emph{feedback}. In particular, the encoder is specified by a
\emph{quantization policy} $\Pi$, which is a sequence of functions
$\{\eta_t\}_{t\ge 0}$ with $\eta_t: \sM^t \times (\sM')^t \times \sX^{t+1} \to
\sM$.  At time $t$, the encoder transmits the $\sM$-valued message
\[
  q_t=\eta_t(I_t), 
\]
where $I_0=X_0$,  $I_t=( q_{[0,t-1]}, q'_{[0,t-1]}, X_{[0,t]})$ for $t \geq 1$,
and  $q'_t$ is the received (noisy) version of $q_t$. The collection of all
such zero delay policies is called the 
set of admissible quantization policies and is denoted by $\Pi_A$.

Upon receiving $q'_t$, the receiver generates the reconstruction, $\hX_t$, also
without delay. A zero delay receiver policy is a sequence of functions
$\gamma=\{\gamma_t\}_{t\ge 0}$ of type $\gamma_t : (\sM')^{t+1} \to \hat{\sX}$,
where $\hat{\sX}$ is the finite reproduction alphabet.  Thus
\[
\hX_t=\gamma_t(q'_{[0,t]})   \quad \text{for all $t\ge 0$.}
\]
Note that, due to the presence of feedback, the encoder  also has access to
$q'_{[0,t]}$ at time $t+1$. The finite and infinite horizon coding problems are
defined analogously to the noiseless case.

The following result is a known extension of Witsenhausen's structure theorem
\cite{Witsenhausen}. 

\begin{thm}[\protect{\cite[Theorem~10.7.1]{YukselBasarBook}}] \label{witsenhausenTheorem2}
  For the problem of transmitting $T$ samples of a  Markov source over a noisy
  channel with feedback, any zero delay quantization policy $\Pi=\{\eta_t\}$ can
  be replaced, without any loss in performance, by a policy
  $\hat{\Pi}=\{\heta_t\}$ which only uses $q'_{[0,t-1]}$ and $X_t$ to generate
  $q_t$, i.e., such that $q_t=\heta_t(q'_{[0,t-1]},X_t)$ for all
  $t=1,\ldots,T-1$.
\end{thm}

Given a quantization policy $\Pi$, for all $t\ge 1$ let $\pi_t
\in {\cal P}(\sX)$ be the  conditional probability defined by
 \[
\pi_t(A)\coloneqq \Pr(X_t\in A | q'_{[0,t-1]})
\]
for any set $A\subset \sX$.

The following result is due to Walrand and Varaiya.

\begin{thm}[\cite{WalrandVaraiya}] \label{WalrandVaraiyaTheorem1} 
 For the problem of transmitting $T$ samples of a  Markov source over a noisy
  channel with feedback, any zero delay
  quantization policy can be replaced, without any loss in performance, by a
  policy which at any time $t= 1, \ldots,T-1$ only uses the conditional
  probability measure $\pi_t=P(dx_t|q'_{[0,t-1]})$ and the state $X_t$ to
  generate $q_t$. In other words, at time $t$ such a policy uses $\pi_t$ to
  select a quantizer $Q_t:\sX \to \sM$ and then $q_t$ is generated as
  $q_t=Q_t(X_t)$.
\end{thm}

Under a Walrand-Varaiya type policy the filtering equation \eqref{filtre} is
modified as
\[
\pi_{t+1}(x_{t+1})= \frac{\sum_{x_t,q_t} \pi_t(x_t)
  T(q'_t|q_t) P(q_t | \pi_t, x_t) 
  P(x_{t+1}|x_t)}{\sum_{x_t,q_t} \sum_{x_{t+1}} \pi_t(x_t) T(q'_t|q_t) P(q_t| \pi_t, x_t) P(x_{t+1}|x_t)}. 
\]
Thus, as before, given $\pi_t$ and $Q_t$, $\pi_{t+1}$ is conditionally
independent of $(\pi_{[0,t-1]},Q_{[0,t-1]})$ and it follows that $\{\pi_t\}$ can
be viewed as $\P(\sX)$-valued controlled Markov process \cite{HernandezLermaMCP}
with $\Q$-valued control $\{Q_t\}$ and average cost up to time $T-1$ given by
\[
E^{\Pi,\gamma}_{\pi_0}\left[\frac{1}{T}   \sum_{t=0}^{T-1} d(X_t,\hX_t)\right] =  E^{\Pi}_{\pi_0}\left[ \frac{1}{T}\sum_{t=0}^{T-1} c(\pi_t,Q_t)\right].\]

The set of deterministic Markov coding policies $\Pi_W$ and deterministic
stationary Markov policies $\Pi_{WS}$ is defined analogously to
Definition~\ref{WVdef}. 

It can be checked that the properties concerning the continuity of the kernel
and the existence of invariant measures apply identically to the new controlled
Markov state pair $(\pi_t, Q_t)$. Under the assumption that $\{X_t\}$ is
irreducible and aperiodic, the simulation argument also applies identically by
considering the same channel noise realizations for both processes $X'_t$ and
$X''_t$; i.e., in the simulation argument we can compare the performance of the
coding schemes by taking the expectations over the channel noise
realizations. Thus, the finite coupling time argument in Lemma
\ref{sim_arg_lemma} applies to this case as well. The following theorem
compactly summarizes the noisy channel analogues of our results in the previos
sections.

\begin{thm}
\begin{itemize}
\item[]
\item[(i)] For the minimization of the finite horizon average distortion
  (\ref{Cost1}), an optimal solution in $\Pi_W$  exists and 
a noisy channel analog of  Proposition~\ref{MeasurableSelectionApplies} holds.

\item[(ii)] For the minimization of the infinite horizon discounted distortion
  (\ref{cost22}), an optimal solution exists and such a solution is  in
  $\Pi_{WS}$, i.e., a noisy channel analog of Proposition~\ref{thm:dc} holds. 

\item[(iii)] The noisy channel version of Theorem~\ref{mainThm} holds: If
  $\{X_t\}$ is irreducible and aperiodic, there exists a policy in $\Pi_{WS}$
  that minimizes the infinite horizon average distortion
  \eqref{infiniteCost}. Furthermore, the convergence rate result
  \eqref{eq:convrate} holds for this optimal policy.

\item[(iv)]  Under the assumption that $\{X_t\}$ is irreducible and aperiodic, if
  $X_0 \sim \pi^*$, where $\pi^*$ is the invariant 
  probability measure, for any  $\epsilon>0$, there exists $K>0$ and  a finite memory,
  nonstationary, but periodic quantization policy with period less than
  $\frac{K}{\epsilon}$ that achieves $\epsilon$-optimal 
  performance Thus the noisy channel version of
  Theorem~\ref{epsOptimalityTheorem} holds. 
\end{itemize}
\end{thm}

\section{Conclusion}
\label{chap:conclu}

Zero delay lossy coding of finite alphabet Markov sources was considered. 
The main result showed that for any irreducible and aperiodic (not necessarily
stationary)  Markov chain 
there exists a stationary and deterministic Markov (Walrand-Varaiya type) policy
that is optimal in the set of zero delay coding policies. 
This result significantly generalizes existing results in
\cite{Witsenhausen}, \cite{WalrandVaraiya}, and \cite{YukLinZeroDelay}.

In addition, it was shown that the distortion of an optimal stationary policy
for time horizon (block length) $T$ converges to the optimal infinite horizon
distortion at a rate $O(1/T)$. As a corollary, the $\epsilon$-optimality of
periodic zero delay codes is established with an explicit bound on the
relationship between $\epsilon$ and the period length.  This result is of
potential practical importance since the code's complexity is directly related to the
length of the period (memory size). Extensions of these results to zero
delay lossy coding over noisy channels with feedback were also
given. 

An interesting open problem is the generalization of the results to continuous
sources such as real or $\R^d$-valued Markov sources. Such a generalization
would be facilitated by an appropriate extension of Lemma \ref{sim_arg_lemma} to
continous alphabets.  Some related results in this direction are available in
\cite{arapostathis2012ergodic}. Another, more challenging open problem of
information theoretic flavor is to find a (preferably) single-letter
characterization of the optimum infinite horizon average distortion of zero
delay coding of Markov sources. As mentioned before, such a characterization is
only known for stationary and memoryless (i.i.d$.$) sources, while for the block
coding problem the distortion rate function gives a (non single-letter)
characterization, and even closed form expressions exist for binary symmetric
Markov sources in a certain range of distortion values \cite{Gra70} as well as
explicit lower and upper bounds \cite{Ber77}.

\appendix[Markov Decision Processes]

\label{sec:controlledmarkovchains}

Let $\sZ$ be a  Borel space (i.e., a Borel subset of  a complete and separable
metric space) and let $\mathcal{P}(\sZ)$ denote the set of all probability
measures on $\sZ$.  

\begin{defn}[Markov Control Model \cite{HernandezLermaMCP}]
A discrete time \emph{Markov control model} (\emph{Markov decision process}) is
a system characterized by the 4-tuple 
\begin{equation*}
(\sZ,\sA, K, c),
\end{equation*}
where
\begin{enumerate}
\item $\sZ$ is the \emph{state space}, the  set of all possible
  states of the system;
\item $\sA$ (a Borel space) is the \emph{control space} (or action space), the
  set of all  controls (actions) $a\in \sA$ that can act 
  on the system;
\item $K=K(\,\cdot\,|z,a)$ is the \emph{transition probability} of the system, a stochastic kernel on
  $\sZ$ given $\sZ\times \sA$, i.e., $K(\,\cdot\,|z,a)$ is a probability measure
  on $\sZ$ for all state-action pairs $(z,a)$, and $K(B|\,\cdot\,\,,\,\cdot\,)$
  is a measurable function from $\sZ\times \sA$ to $[0,1]$ for each Borel set
  $B\subset \sZ$;
\item $c:\sZ \times \sA \to [0,\infty)$ is the \emph{cost per time stage
    function} of the system, a function  $c(x,a)$ of the state and the control.
\end{enumerate}
\end{defn}

Define the \emph{history} spaces $\sH_t$ at time $t\ge 0$ of the Markov control
model by  $\sH_0 \coloneqq \sZ$ and $\sH_t \coloneqq (\sZ\times \sA)^t  
 \times \sZ$.  Thus a specific history  $h_t \in \sH_t$ has
the form $h_t = (z_0,a_0,\ldots,z_{t-1},a_{t-1},z_t)$.

\begin{defn}[Admissible Control Policy \cite{HernandezLermaMCP}]
  An \emph{admissible control policy} $\Pi = \{\alpha_t\}_{t \geq 0}$, also
  called a \emph{randomized control policy} (more simply a \emph{control policy}
  or a \emph{policy}) is a sequence of stochastic kernels on the action space
  $\sA$ given the history $\sH_t$.  The set of all randomized control policies is
  denoted by $\Pi_A$. A \emph{deterministic policy} $\Pi$ is a sequence of
  functions $\{\alpha_t\}_{t \geq 0}$, $\alpha_t : \sH_t \to \sA$, that
  determine the control used at each time stage deterministically, i.e., $a_t =
  \alpha_t(h_t)$. The set of all deterministic policies is denoted
  $\Pi_D$.  Note that $\Pi_D \subset \Pi_A$. A \emph{Markov policy} is a policy
  $\Pi$ such that for each time stage the choice of control only depends on the
  current state $z_t$, i.e.,\ $\Pi = \{\alpha_t\}_{t \geq 0}$ with  $\alpha_t
  : \sZ \to \mathcal{P}(\sA)$.  The set of all Markov policies is denoted by
  $\Pi_M$. The set of deterministic Markov policies is denoted by
  $\Pi_{MD}$. A \emph{stationary policy} is a Markov policy $\Pi =
  \{\alpha_t\}_{t \geq 0}$ such that $\alpha_t = \alpha$ for all $t \geq 0$ for some
   $\alpha : \sZ \to \mathcal{P}(\sA)$.  The set of all stationary policies
  is denoted by $\Pi_S$ and the set of deterministic stationary policies is denoted
  by $\Pi_{SD}$. 
\end{defn}

According to the Ionescu Tulcea theorem (see \cite{HernandezLermaMCP}), the
transition kernel $K$,  an
initial probability distribution $\pi_0$ on $\sZ$,  and a policy $\Pi$ define a unique
probability measure $P_{\pi_0}^{\Pi}$ on
$\sH_{\infty}=(\sX\times\sA)^{\infty}$, the distribution of the state-action
process $\{(Z_t,A_t)\}_{t\ge 0}$. The resulting state process $\{Z_t\}_{t\ge 0}$
is called a \emph{controlled Markov process}. 
The expectation with respect to $P_{\pi_0}^{\Pi}$ is denoted by $E_{\pi_0}^{\Pi}$.
If $\pi_0=\delta_z$, the point mass at $z\in \sZ$, we write $P_{z}^{\Pi}$ and $E_{z}^{\Pi}$ instead of
$P_{\delta_z}^{\Pi}$ and $E_{\delta_z}^{\Pi}$.

In an \emph{optimal control problem}, a performance objective $J$ of the system
is given and the goal is to find the controls that minimize (or maximize) that
objective.  Some common optimal control problems for Markov control models are
the following:

\begin{enumerate}
\item \emph{Finite Horizon Average Cost Problem}: Here the goal is to find policies that minimize the average
cost
\begin{equation}\label{Cost11}
J_{\pi_0}(\Pi,T) \coloneqq E^{\Pi}_{\pi_0}\left[\frac{1}{T} \sum_{t=0}^{T-1} c(Z_t,A_t)\right],
\end{equation}
for some $T \ge 1$.

\item \emph{Infinite Horizon Discounted Cost Problem}:
Here the goal is to find policies that minimize
\begin{equation}
J^{\beta}_{\pi_0}(\Pi) \coloneqq \lim_{T \to \infty} E^{\Pi}_{\pi_0}\left[ \sum_{t=0}^{T-1} \beta^t c(Z_t,A_t)\right],
\end{equation}
for some $\beta \in (0,1)$.

\item \emph{Infinite Horizon Average Cost Problem}:
In the more challenging infinite horizon control problem  the goal is to find
policies that minimize the average cost 
\begin{equation}
J_{\pi_0}(\Pi) \coloneqq \limsup_{T \to \infty} E^{\Pi}_{\pi_0}\left[\frac{1}{T}   \sum_{t=0}^{T-1} c(Z_t,A_t)\right]. \label{infiniteCost1}
\end{equation}
\end{enumerate}

The Markov control model together with the performance objective is called a
\emph{Markov decision process}. 

A common method to solving finite horizon Markov control problems is by
\emph{dynamic programming}, which involves working backwards from the final time
stage to solve for the optimal sequence of controls to use.  The optimality of
this algorithm is guaranteed by Bellman's principle of optimality.

\begin{thm}[Bellman's Principle of Optimality \protect{\cite[Chapter 3.2]{HernandezLermaMCP}}]
\label{thm:bellman}
Given a finite time horizon $T\ge 1$, define a sequence of functions
$J_T,\ldots,J_0$ on 
$\sZ$ recursively such that
\begin{equation*}
J_{T}(z_T) \equiv 0,
\end{equation*}
and for $0\leq t < T$ and $z\in \sZ$,
\begin{equation}
\label{bellmanthmeqn}
J_t(z) \coloneqq \min_{a\in \sA} \left[ 
c(z,a) + \int_{\sZ} J_{t+1}(z')
  K(dz' | z ,a) \right].
\end{equation}
If the $J_t$ are measurable and there exist measurable $f_t: \sZ\to \sA$ 
such that $a=f_t(z)$ achieves the above minimum for all $t=0,\ldots,T-1$, then
the deterministic Markov policy $\Pi \coloneqq (f_0,\ldots,f_{T-1})$ is optimal
with cost $J_{z_0}(\Pi,T) = J_0(z_0)$. 
\end{thm}

Quite general conditions exist under witch the two  assumptions of the above theorem
hold  \cite[Chapter 3.3]{HernandezLermaMCP}.

For the infinite horizon discounted cost Markov control problem, one can also
use an iteration algorithm to obtain an optimal policy.  This approach is
commonly called the \emph{successive approximations} or value iteration  method
\cite[Chapter 4.2]{HernandezLermaMCP}.

A stochastic kernel $K$ on $\mathbb{\sZ}$ given $\sZ\times \sA$ is called weakly
continuous if the function $(a,z)\mapsto \int_{\sZ} v(z') K(dz'|z,a)$ is
continuous whenever $v$ is a \emph{bounded and continuous} real function on
$\sZ\times \sA$. It is called strongly continuous if 
the $(a,z)\mapsto \int_{\sZ} v(z') K(dz'|z,a)$ is continuous whenever $v$ is a
\emph{measurable and bounded} real function on $\sZ\times \sA$. The next
theorem follows from \cite[Chapter 8.5]{HernandezLermaMCP2}. 

\begin{thm} \label{thm:iterative}   Suppose the following conditions hold:
\begin{itemize}
\item[(i)] The one stage cost $c$ is continuous, nonnegative, and bounded;
\item[(ii)]  $\sA$ is compact;
\item[(iii)] the  transition kernel $K$
  is weakly continuous. 
\end{itemize}
Then for any $\beta \in (0,1)$, the pointwise limit
  $J(z)$ as $t\to \infty$, of the sequence defined by
\begin{equation*}
J_t(z) = \min_{a \in \sA} \left[ c(x, a) + \beta \int_{\sZ} J_{t-1}(z')
  K(z'|z,a) \right], 
\quad z \in \sZ,
\end{equation*}
with $J_0(z) \equiv 0$, yields the optimum cost in the
  infinite horizon discounted cost problem (i.e., $\inf_{\Pi\in \Pi_A}
  J_{z}^{\beta}= J(z)$. 
Furthermore, there exists a measurable function $f:\sZ\to \sA$ such that 
\[
J(z) = c(x, f(z)) + \beta \int_{\sZ} J_{t-1}(z')
  K(z'|z,f(z))
\]
and the policy $\Pi=\{f\}$ is an optimal stationary Markov policy. 
\end{thm}

Finally, for the infinite horizon average cost Markov control problem, we
give a brief overview of the \emph{average cost optimality equation} (ACOE). 
When the ACOE holds for a deterministic and stationary Markov policy 
$\Pi$, we know $\Pi$ is optimal for the infinite horizon average cost
problem.

\begin{defn}
Let $h$ and $g$  be  measurable real functions on $\sZ$ and let  $f:\sZ\to \sA$ be
measurable. Then $(g,h,f)$ is said to be a canonical triplet if for all $z\in\sZ$, 
\begin{align}
g(z)&=\inf_{a\in \sA} \int_{\sZ} g(z') K(dz'|z,a) \label{eq:can1} \\
g(z) + h(z) &= \inf_{a\in \sA} \left(c(z,a) + \int_{\sZ} h(z') K(dz'|z,a)
              \right)  \label{eq:can2}
\end{align}
and
\begin{align}
g(z) &= \int_{\sZ} g(z') K(dz'|z,f(z))  \label{eq:can3} \\
g(z) + h(z) &= c(z,f(z)) + \int_{\sZ} h(z') K(dz'|z,f(z)). \label{eq:can4}
\end{align}
\end{defn}

Equations \eqref{eq:can1}--\eqref{eq:can4} are called the \emph{canonical
  equations}. In case $g$ is a constant,  $g\equiv g^*\in [0,\infty)$, these
equations reduce to
\begin{align}
g^* + h(z) &= \inf_{a\in \sA} \left(c(z,a) + \int_{\sZ} h(z') K(dz'|z,a)
              \right)  \label{eq:ACOE1} \\
g^* + h(z) &= c(z,f(z)) + \int_{\sZ} h(z') K(dz'|z,f(z)) \label{eq:ACOE2}
\end{align}
and \eqref{eq:ACOE1} is called the \emph{average cost optimality equation}
(ACOE). 

The ACOE is of central importance in the theory of
infinite horizon average cost problems since (as can be shown
\cite[Chapter 5.2]{HernandezLermaMCP}), with the 
additional condition that $\limsup_{T\to \infty} (1/T) E_{z_0}^{\Pi} \big[
h(Z_T)] =0$ for all $z_{0}\in \sZ$ and $\Pi\in \Pi_A$, it implies that
the deterministic and stationary Markov policy $\Pi^*=\{f\}$ is optimal in $\Pi_A$
and $g^*$ is the value function, i.e.,
\[
  g^* = J_{z_0}(\Pi^*) = \inf_{\Pi\in \Pi_A}  J_{z_0}(\Pi). 
\]

Although several general sufficient conditions for the ACOE to hold exist (see,
e.g., Assumptions 4.2.1 and 5.5.1 in \cite{HernandezLermaMCP}), these conditions
are restrictive in our setup since they involve the strong continuity of the
transition kernel. In our results we take $\sZ$ to be the space
probability measures, which makes strong continuity  too strict a condition  in
general \cite{saldi2014near} \cite{FeKaZg14}. More relaxed conditions that involve weak
continuity are available in the literature, see \cite{GoHe95}
\cite{Veg03}. Since for us it is enough to consider compact state and action spaces
and uniformly bounded cost, the following theorem will suffice.  Recall
that
\[
J^{\beta}_{z} =\inf_{\Pi\in \Pi_A} J_z^{\beta}(\Pi). 
\]

\begin{thm}[\protect{\cite[Theorem 3.3]{saldi2014near}}]
\label{thm3}
Suppose conditions (i)--(iii) of Theorem~\ref{thm:iterative} hold and, in
addition,
\begin{itemize}
\item[(iv)] the state space $\sZ$ is compact;
\item[(v)] the family of functions  $\{h_{\beta}: \beta\in (0,1)  \}$, with
\[
h_{\beta}(z) = J^{\beta}_{z} - J^{\beta}_{z_0}
\]
for some fixed $z_0\in \sZ$, is uniformly bounded and equicontinuous. 
\end{itemize}
Then there exist a constant $g^*\ge 0$, a continuous and
bounded function $h:\sZ\to \mathbb{R}$, and a measurable function
$f^*:\sZ\to \sA$ such that $(g^{*},h,f^{*})$ is a canonical triplet that
satisfies the ACOE. Thus the deterministic and stationary Markov policy
$\Pi^*=\{f^{*}\}$ is optimal in $\Pi_A$ and $g^*$ is the value function, i.e.,
\begin{align*}
g^* =  J_{z_0}(\Pi^*) = \inf_{\Pi\in \Pi_A} \limsup_{T \to \infty} E^{\Pi}_{z_0}\left[\frac{1}{T}   \sum_{t=0}^{T-1} c(Z_t,A_t)\right],  
\end{align*}
for all $z_0 \in \sZ $.
\end{thm}

Recall the definition
\[
J_{\pi_0}(\Pi,T) \coloneqq E^{\Pi}_{\pi_0}\left[\frac{1}{T} \sum_{t=0}^{T-1} c(Z_t,A_t)\right].
\]
For the canonical triplet  $(\rho^{*},h,f^{*})$ in the preceding theorem,
\cite[p.~80]{HernandezLermaMCP} shows that for all $z_0\in \sZ$ and $T\ge 1$,
\begin{equation}
\label{eq:kbound}
J_{z_0}(\Pi,T) = g^* + \frac{1}{T} \bigg(  h(z_0) - E_{z_0}^{\Pi^*} h(Z_T) \bigg).
\end{equation}
Also, the function $h$ in Theorem~\ref{thm3} is the pointwise limit of the
sequence $\{h_{\beta_n}(z)\}$ along some sequence of discount factors
$\{\beta_n\}$ such that $\lim_{n\to \infty} \beta_n= 1$. Thus if
$\{h_{\beta_n}(z)\}$ is uniformly bounded, say  
$|h_{\beta}(z)| \leq K/2$  for all $z\in \sZ$ and $\beta\in (0,1)$, then
$|h(z)|\le K/2$ for all $z$,  and so \eqref{eq:kbound} implies 
\begin{equation}
\label{eq:convrate-triplet}  
J_{z_0}(\Pi^*,T) - g^* = J_{z_0}(\Pi^*,T)- J_{z_0}(\Pi^*) \le \frac{K}{T}
\end{equation}
for all $T\ge 1$.

For further details on controlled Markov processes, see \cite{HernandezLermaMCP}.

%
%

\bibliographystyle{ieeetr}

\end{document}